\documentclass[english,11pt]{elsarticle}
\usepackage[T1]{fontenc}
\usepackage[utf8]{inputenc}
\usepackage{amsthm}
\usepackage{amsmath}
\usepackage{amssymb}
\usepackage{esint}

\makeatletter
  \theoremstyle{definition}
  \newtheorem{defn}{\protect\definitionname}
  \theoremstyle{plain}
  \newtheorem{prop}{\protect\propositionname}

\makeatother

\usepackage{babel}
  \providecommand{\definitionname}{Definition}
  \providecommand{\propositionname}{Proposition}

\begin{document}

\title{Some properties of generalized Fisher information in the context
of nonextensive thermostatistics\footnote{This is a preprint version that differs from the published version, Physica A, vol. 392, issue. 15, pp. 3140-3154, doi:10.1016/j.physa.2013.03.062, in some improvements after refereeing, pagination and typographics details.}}

\author{J.-F. Bercher}

\ead{jf.bercher@esiee.fr}

\address{Laboratoire d'informatique Gaspard Monge, UMR 8049 \\ 
ESIEE-Paris, Université Paris-Est\\
 5 bd Descartes, 77454 Marne la Vallée Cedex 2, France\\
 tel: 33-1-45-92-65-15 fax: 33-1-45-92-66-99}
\begin{abstract}
We present two extended forms of Fisher information that fit well
in the context of nonextensive thermostatistics. We show that there exists
an interplay between these generalized Fisher information, the generalized
$q$-Gaussian distributions and the $q$-entropies. The minimum of
the generalized Fisher information among distributions with a fixed
moment, or with a fixed $q$-entropy is attained, in both cases, by
a generalized $q$-Gaussian distribution. This complements the fact
that the $q$-Gaussians maximize the $q$-entropies subject to a moment
constraint, and yields new variational characterizations of the generalized
$q$-Gaussians. We show that the generalized Fisher information naturally
pop up in the expression of the time derivative of the $q$-entropies,
for distributions satisfying a certain nonlinear heat equation. This
result includes as a particular case the classical de Bruijn identity.
Then we study further properties of the generalized Fisher information
and of their minimization. We show that, though non additive, the generalized Fisher information of a 
combined system is upper bounded. 
In the case of mixing, we show that the generalized Fisher information
is convex for $q\geq1.$ Finally, we show that the minimization of
the generalized Fisher information subject to moment constraints satisfies
a Legendre structure analog to the Legendre structure of thermodynamics. \end{abstract}
\begin{keyword}
Generalized Fisher information \sep Generalized Rényi and Tsallis
entropies \sep Generalized $q$-gaussian distributions \sep Information
theoretic inequalities

\PACS {02.50.-r} \sep {05.90.+m} \sep {89.70.+c} 

\end{keyword}
\maketitle
\date{Typesetted \today}

\section{Introduction and preliminary definitions}

Information measures are important both for the foundation of information
sciences and for practical applications of information processing.
They are of outmost importance in several areas of physics, beginning
of course with statistical physics. Entropies and Fisher information
have been identified as useful and versatile tools for characterizing
complex systems, see e.g. \cite{dehesa_Cramer-Rao_2006,dehesa_fisher_2006,dehesa_generalized_????}.
Classical information theoretic inequalities, as described for instance
in \cite{cover_elements_2006,dembo_information_1991}, interrelate
information measures. These inequalities have proved to be useful
for communication theoretic problems and engineering applications.
They are also connected to uncertainty relations in physics, see e.g.
\cite{stam_inequalities_1959,folland_uncertainty_1997,sen_statistical_2011}, 
and to functional inequalities in mathematics.
The aim of the paper is to present and study two generalized Fisher
information measures that fit well in the context of nonextensive thermostatistics,
and to describe the interplay between information measures and generalized
$q$-Gaussians, thus providing new characterizations of these generalized
$q$-Gaussians. We first recall the context, our notations and main
definitions.

\subsection*{Generalized entropies}

The Boltzmann entropy, linked to the Shannon entropy of information
theory, is at the heart of thermodynamics. For energy constraints,
it is well-known that the maximum entropy distribution is the Gaussian
distribution. For the analysis of complex systems, generalized entropies,
which reduces to the standard one as a particular case, have been
proposed. In particular, the nonextensive thermodynamics derived from Tsallis entropy,
 has received a high attention. There is a wide variety of applications where experiments, numerical
results and analytical derivations fairly agree with the new formalisms
\cite{tsallis_introduction_2009}. Some physical applications of
the generalized entropies, --\,\,including statistics of cosmic
rays, defect turbulence, optical lattices, systems with long-range
interactions, superstatistical systems, etc., can be found in the
recent review \cite{beck_generalised_2009} and references therein.
Let us recall here that if $x$ is a vector of $\Omega\subseteq\mathbb{R}^{n}$,
and $f(x)$ a probability density defined with respect to the Lebesgue
measure, then following \cite{golomb_information_1966}, the information
generating function is the quantity 
\begin{equation}
M_{q}[f]=\int_{\Omega}f(x)^{q}\mathrm{d}x,
\end{equation}
for $q\geq0$. This quantity is also sometimes called ``entropic
moment''. The Tsallis and Rényi entropies are respectively given by
\begin{eqnarray}
S_{q}[f] &=\frac{1}{1-q}\left(M_{q}[f]-1\right), \\
H_{q}[f] &=\frac{1}{1-q}\log M_{q}[f].
\end{eqnarray}
For our purposes, it is not always necessary to distinguish between
the two entropies that we will call collectively $q$-entropies. Finally,
we will also call ``entropy power'' of order $q$, or $q$-entropy
power the quantity 
\begin{equation}
N_{q}[f]=\exp\left(\frac{2}{n}H_{q}[f]\right)=M_{q}[f]^{\frac{2}{n}\,\frac{1}{1-q}}=\left(\int_{\Omega}f(x)^{q}\text{d}x\right)^{\frac{2}{n}\,\frac{1}{1-q}},\label{eq:defEntropyPower}
\end{equation}
for $q\neq 1$. For $q=1,$ we let $N_{q}[f]=\exp\left(\frac{2}{n}H_{1}[f]\right),$
where $H_{1}[f]$ is the Boltzmann-Shannon entropy.

\subsection*{Generalized $q$-Gaussian distributions}

In the context of nonextensive thermostatistics, the role devoted
to the standard Gaussian distribution is extended to generalized $q$-Gaussian
distributions, which include the standard Gaussian as a special case.
These generalized $q$-Gaussian distributions form a versatile family
that can describe problems with compact support as well as problems
with heavy tailed distributions. They are also analytical solutions
of actual physical problems, see e.g. \cite{lutz_anomalous_2003,schwaemmle_q-gaussians_2008,vignat_isdetection_2009}, and are sometimes known as Barenblatt-Pattle
functions, following their identification by \cite{barenblatt_unsteady_1952,pattle_diffusion_1959}.
We shall also mention that the generalized $q$-Gaussian distributions
appear in other fields, namely as the solution of non-linear diffusion
equations, or as the distributions that saturate some sharp inequalities
in functional analysis \cite{del_pino_best_2002,del_pino_optimal_2003,cordero-erausquin_mass-transportation_2004,agueh_sharp_2008}. 
\begin{defn}
Let $x$ be a random vector of $\mathbb{R}^{n},$ and let $|x|$ denote
its Euclidean norm. For $\alpha\in(0,\infty),$ $\gamma$ a real positive
parameter and $q>(n-\alpha)/n,$ the generalized Gaussian with parameter
$\gamma$ has the radially symmetric probability density 
\begin{equation}
G_{\gamma}(x)=\begin{cases}
\frac{1}{Z(\gamma)}\left(1-\left(q-1\right)\gamma|x|^{\alpha}\right)_{+}^{\frac{1}{q-1}} & \text{for }q\not=1\\
\frac{1}{Z(\gamma)}\exp\left(-\gamma|x|^{\alpha}\right) & \text{if }q=1
\end{cases}\text{ }\label{eq:qgauss_general}
\end{equation}
where we use the notation $\left(x\right)_{+}=\mbox{max}\left\{ x,0\right\} $,
and where $Z(\gamma)$ is the partition function such that $G_{\gamma}(x)$
integrates to one. Its expression is given by 
\begin{equation}
Z(\gamma)=\frac{1}{\alpha}\left(\gamma\right)^{-\frac{n}{\alpha}}n\,\omega_{n}\times\begin{cases}
(1-q)^{-\frac{n}{\alpha}}B\left(\frac{n}{\alpha},-\frac{1}{q-1}-\frac{n}{\alpha}\right) & \text{for }1-\frac{\alpha}{n}<q<1\\
(q-1)^{-\frac{n}{\alpha}}B\left(\frac{n}{\alpha},\frac{1}{q-1}+1\right) & \text{for }q>1\\
\Gamma\left(\frac{n}{\alpha}\right) & \text{if }q=1
\end{cases}
\end{equation}
where $B(x,y)$ denotes the beta function and $\omega_{n}=\pi^{\frac{n}{2}}/\Gamma(\frac{n}{2}+1)$
is the volume of the $n$-dimensional unit ball.  The parameter $\gamma$
is linked to the moments of the density. For instance, if $m_{\alpha}[G_{\gamma}]$
denotes the moment of order $\alpha,$ we have $\gamma^{-1}=\left((1+\frac{\alpha}{n})q-1\right)m_{\alpha}[G_{\gamma}]$, 
$\text{for }q>n/(n+\alpha)$. We may also note that $\gamma^{-\frac{1}{\alpha}}$
is a scale parameter, e.g. we have $Z(\gamma)=\gamma^{-\frac{1}{\alpha}}Z(1).$
Usually, the term $q$-Gaussian corresponds to the case $\alpha=2$
above. The generalized $q$-Gaussians are sometimes called stretched
$q$-Gaussians, where $\alpha$ is a stretch parameter. 
\end{defn}
In the sequel, we will denote $G$ the generalized Gaussian obtained
with $\gamma=1.$ For $q>1$, the density has a compact support, while
for $q\leq1$ it is defined on the whole $\mathbb{R}^{n}$ and behaves
as a power distribution for $|x|\rightarrow\infty.$ A shorthand notation
for the expression of the generalized $q$-Gaussian density is 
\begin{equation}
G_{\gamma}(x)=\frac{1}{Z(\gamma)}\exp_{q^{*}}\left(-\gamma|x|^{\alpha}\right),\label{eq:DefQGaussian}
\end{equation}
with $q^{*}=2-q$, and where the so-called $q$-exponential function,
and its inverse the $q$-logarithm, are given by
\begin{alignat}{1}
\exp_{q}(x):=\left(1+(1-q)x\right)_{+}^{\frac{1}{1-q}} & \text{ for }q\neq1\text{ and }\exp_{q=1}(x):=\exp(x),\label{eq:defExpq-1}\\
\ln_{q}(x):=\frac{x^{1-q}-1}{1-q} & \text{ for }q\neq1\text{ and }\ln_{q=1}(x):=\ln(x).\label{eq:defLnq-1-1}
\end{alignat}

It is well known that the Gaussian distribution is a central distribution
with respect to classical information measures and inequalities. In
particular, the Gaussian distribution is both a maximum entropy and
a minimum Fisher information distribution over all distributions with
the same variance. We will see that the same kind of result holds
for the family of generalized $q$-Gaussians, for Rényi or Tsallis
entropy and suitable extensions of the Fisher information.

\subsection*{Generalized ($\beta,q$)-Fisher information}

In the context of nonextensive thermostatistics, several authors have
introduced and studied generalized versions of the Fisher information.
Among these contributions we note the series of papers \cite{chimento_naudts-like_2000,casas_fisher_2002,pennini_rnyi_1998,pennini_semiclassical_2007},
and the proposals of Naudts \cite{naudts_generalised_2008,naudts_q-exponential_2009}
and Furuichi \cite{furuichi_maximum_2009,furuichi_generalized_2010},
which are related to the generalized Fisher information that are used
here -- see \cite{bercher_generalized_2012} for a discussion of
these links. 

In our recent work \cite{bercher_generalized_2012,bercher_generalized_2012b},
we have thrown a bridge between concepts in estimation theory and
tools of nonextensive thermostatistics. Using the notion of escort
distribution, we have introduced a generalized Fisher information
and established an extended version of the Cramér-Rao inequality for
parameter estimation. In the case of a location parameter, it reduces
to an extended version of the standard Cramér-Rao inequality, which
is saturated by the generalized $q$-Gaussians.  A closely related
generalized Fisher information has been originally introduced by Lutwak
\textit{et al. }\cite{lutwak_cramer_2005} in information theory
and extended to the multidimensional in \cite{lutwak_extensions_2012}
and independently in \cite{bercher__2012}. These two extensions
of the Fisher information measure are defined as follows. 
\begin{defn}
Let $f(x)$ be a probability density function defined over a subset
$\Omega$ of $\mathbb{R}^{n}$. Let $|x|$ denote the Euclidean norm
of $x$ and $\nabla f$ the gradient operator. If $f(x)$ is continuously
differentiable over $\Omega,$ then for $q\geq0$, $\beta>1$, the
generalized $(\beta,q)$-Fisher information is defined by
\begin{align}
\phi_{\beta,q}[f] & =\int_{\Omega}f(x)^{\beta(q-1)+1}\left(\frac{|\nabla f(x)|}{f(x)}\right)^{\beta}\mathrm{d}x\label{eq:Phi_GenFishera}\\
 & =E\left[f(x)^{\beta(q-1)}\left|\nabla\ln f(x)\right|^{\beta}\right]=\, E\left[\left|\nabla\ln_{q*}f(x)\right|^{\beta}\right]\label{eq:Phi_GenFisherb}
\end{align}
with $q_{*}=2-q.$

With the same notations and assumptions as above, the second generalized
$(\beta,q)$-Fisher information is 
\begin{align}
I_{\beta,q}\left[f\right] & =\left(\frac{q}{M_{q}\left[f\right]}\right)^{\beta}\,\phi_{\beta,q}\left[f\right]=\left(\frac{q}{M_{q}\left[f\right]}\right)^{\beta}\, E\left[\left|\nabla\ln_{q*}f(x)\right|^{\beta}\right].\label{eq:I_GenFishera}
\end{align}

\end{defn}
These two extensions of the Fisher information measure, the $(\beta,q)$-Fisher
information $\phi_{\beta,q}[f]$ and $I_{\beta,q}[f]$, depend on
an entropic index $q$ and on a parameter $\beta.$ As mentioned above,
the $\phi_{\beta,q}[f]$  form has been introduced by Lutwak \textit{et al. }\cite{lutwak_cramer_2005}
in information theory and is related to general results in functional
analysis. The second form of generalized Fisher information has been
introduced in the derivation of a general Cramér-Rao inequality for parameter estimation \cite{bercher_generalized_2012,bercher_generalized_2012b}.
These two information only differ by a prefactor, but actually originates
from two different settings and lead to different, though similar,
inequalities and characterizations. Of course, both information reduce to the
standard Fisher information in the case $q=1,$\,\,$\beta=2.$ We
will see that these two generalized $(\beta,q)$-Fisher information
interplay nicely with $q$-entropies and generalized $q$-Gaussians,
generalizing classical information relations, and thus allow a natural
extension of the usual Shannon-Fisher-Gaussian setting.

\subsection*{Structure and contributions of the paper}

As already indicated, the goal of the paper is to introduce two generalized
Fisher information in the context of nonextensive thermostatistics,
to document their main properties and study the interplay between
generalized information measures and generalized $q$-Gaussians. Though
the paper reviews some of our previous results, most of the results
presented here are new. The paper is organized into two main parts.
In the first part, we present several variational characterizations
of the generalized $q$-Gaussian distributions, where the generalized
Fisher information play a fundamental role. The new findings include
(i) the second generalized Stam inequality (\ref{eq:GeneralizedStamInequalityforI})
which lower bounds the product of the entropy power and of the extended
Fisher information, (ii) the definition of several information functionals
minimized by the generalized $q$-Gaussians, (iii) the derivation
of the extended de Bruijn inequality (\ref{eq:ExtendedDeBruijn},
\ref{eq:ExtendedDeBruijnb}) and of the de Bruijn entropy power inequality
(\ref{eq:ExtendeddeBruijnEntropyPower}) which intimately links the
$q$-entropies to the generalized Fisher information through a doubly
nonlinear diffusion equation. In the second section, we establish
several properties of the generalized Fisher information which might
be useful for thermodynamics considerations. The discussion includes
the additivity and mixing properties of the generalized Fisher information,
as well as an original derivation of the preservation of the Legendre
structure.

\section{Variational characterizations of generalized $q$-Gaussians and generalized
Fisher information}

It is known that the $q$-Gaussians maximize the Tsallis or
Rényi entropy among all probability distributions with a given variance
or covariance \cite{johnson_results_2007}. This means that the $q$-Gaussians
are the canonical distributions of the generalized thermostatistics
based on Tsallis or Rényi entropy. The maximum entropy formalism directly
induces a variational characterization of generalized $q$-Gaussians
as the extremal functions of the Lagrangian 
\[
L[f]=H_{q}[f]+\mu m_{\alpha}[f],
\]
where $\mu$ is a Lagrange multiplier and $m_{_{\alpha}}[f]$ denotes
the moment of order $\alpha$ of the probability distribution $f.$
A remarkable point is that the generalized $q$-Gaussians are also
solutions of other variational problems as well, and that these problems
involve the generalized ($\beta,q$)-Fisher information we defined
above. Such characterizations are described now.

\subsection{The minimum of generalized Fisher information among distributions
with a given $q$-entropy is attained for generalized $q$-Gaussians}

Another characterization of the generalized $q$-Gaussian is indeed
the fact that they minimize the extended Fisher information among
all distributions with a given $q$-entropy. This result is a consequence
of a generalized Stam inequality which lower bounds a product of the
entropy power and of the generalized Fisher information. This inequality
is stated in the Proposition below and has been established in \cite{lutwak_cramer_2005}
in the monodimensional case, in \cite{lutwak_extensions_2012} and
\cite{bercher__2012} in the multidimensional case, and in \cite{bercher_generalized_2012b}
for arbitrary norms. The statement of case (b) is new. 
\begin{prop}
{[}Generalized Stam inequalities{]} For $n\geq1,$ $\beta$ and $\alpha$
Hölder conjugates of each other, $\alpha>1,$ and $q>\max\left\{ (n-1)/n,\, n/(n+\alpha)\right\} $,
then for any probability density on $\mathbb{R}^{n}$, supposed continuously
differentiable, the following generalized Stam inequalities hold
\begin{alignat}{1}
\text{(a)\,\,\,\,\,\,} & \phi_{\beta,q}\left[f\right]^{\frac{1}{\beta}}\, N_{q}[f]^{\frac{\lambda}{2}}\geq\phi_{\beta,q}\left[G\right]^{\frac{1}{\beta}}\, N_{q}[G]^{\frac{\lambda}{2}},\label{eq:GeneralizedStamInequality}\\
\text{(b)\,\,\,\,\,\,} & I_{\beta,q}\left[f\right]^{\frac{1}{\beta}}\, N_{q}[f]^{\frac{1}{2}}\geq\, I_{\beta,q}\left[G\right]^{\frac{1}{\beta}}\, N_{q}[G]^{\frac{1}{2}}.\label{eq:GeneralizedStamInequalityforI}
\end{alignat}
with 
\begin{equation}
\lambda=n(q-1)+1>0\label{eq:def_lambda}
\end{equation}
 and with equality if and only if $f$ is any generalized $q$-Gaussian
(\ref{eq:DefQGaussian}).
\end{prop}
In the previous expressions, $\phi_{\beta,q}\left[G\right]$, $I_{\beta,q}\left[G\right]$
and $N_{q}[G]$ are the values taken by the ($\beta,q$)-Fisher information
and the entropy power when the probability density $f$ is the generalized
$q$-Gaussian $G$. The exact expressions of these quantities are
given in the Appendix, section \ref{sec:Information-measures-of}.
\begin{proof}
The inequality (a) is proved in \cite{lutwak_extensions_2012} and
\cite{bercher__2012} using a general sharp Gagliardo-Nirenberg inequality
due to Cordero et al. \cite{cordero-erausquin_mass-transportation_2004}.
It is quite easy to get the same kind of generalized Stam inequality
for the second generalized Fisher information. Indeed, for the inequality
(b), it suffices to replace $\phi_{\beta,q}\left[f\right]^{\frac{1}{\beta}}$
in (\ref{eq:GeneralizedStamInequality}) by $\, I_{\beta,q}\left[f\right]^{\frac{1}{\beta}}\, M_{q}\left[f\right]/q$,
to use the fact that $M_{q}[f]=N_{q}[f]^{\frac{n}{2}(1-q)}$ and to
simplify the exponent to obtain (\ref{eq:GeneralizedStamInequalityforI}).
\end{proof}
The generalized Stam inequalities imply that the generalized $q$-Gaussian
minimize the extended Fisher information within the set of probability
distributions with a fixed $q$-entropy power: $G_{\gamma}$ is the
solution of 
\begin{equation}
\inf_{f}\left\{ \phi_{\beta,q}[f]\text{ or }I_{\beta,q}[f]:\, f\in\mathcal{P},\, N_{q}[f]=N_{q}[G_{\gamma}]\right\} ,\label{eq:pb_min_Fish_Nq}
\end{equation}
where $\mathcal{P}$ is the probability simplex. In turn, this leads
to the following variational description.
\begin{prop}
\label{Prop2}The generalized $q$-Gaussian $G$ minimize the following
two information functionals
\begin{align}
J_{1}[f] & =\frac{1}{\beta\lambda}\phi_{\beta,q}[f]N_{q}[G]+\frac{1}{2}\phi_{\beta,q}[G]N_{q}[f]\label{eq:J1_informationfunctional}\\
J_{2}[f] & =\frac{1}{\beta}I_{\beta,q}[f]N_{q}[G]+\frac{1}{2}I_{\beta,q}[G]N_{q}[f]\label{eq:eq:J2_informationfunctional}
\end{align}
with $\lambda$ given by (\ref{eq:def_lambda}). 
\end{prop}
Another statement of this result could be in terms of an energy functional.
 Let $f=u^{k},$ with $k=\beta/\left(\beta(q-1)+1\right)$. With
this notation, the generalized Fisher information reduces to
\begin{equation}
\phi_{\beta,q}[f]=|k|^{\beta}\int_{\Omega}|\nabla u(x)|^{\beta}\mathrm{d}x,\label{eq:GenFisherDefinition_reduced}
\end{equation}
which is the $\beta$-Dirichlet energy of function $k\, u(x)$. In
addition, $M_{q}[f]=\int_{\Omega}u(x)^{kq}\mathrm{d}x$ is a measure
of informational energy of order $kq$ in the sense of \cite{pardo_order-alpha_1986,onicescu_energie_1966}.
Then, (\ref{eq:J1_informationfunctional}) becomes
\begin{equation}
J_{1}[f]=\frac{1}{\beta\lambda}N_{q}[G]\,\int_{\Omega}|\nabla u(x)|^{\beta}\mathrm{d}x+\frac{1}{2}\phi_{\beta,q}[G]\left(\int_{\Omega}u(x)^{kq}\mathrm{d}x\right)^{\frac{1}{1-q}}
\end{equation}
 a weighted sum of these two energies. 
\begin{proof}
Consider the first generalized Stam inequality. As already indicated,
a direct consequence of this inequality is that the generalized $q$-Gaussian $G_\gamma$
solves problem (\ref{eq:pb_min_Fish_Nq}). The Lagrangian corresponding
to this problem is 
\begin{equation}
J(f)=\phi_{\beta,q}[f]+\mu N_{q}[f].
\end{equation}
This Lagrangian is minimum for a $q$-Gaussian $G_{\gamma}$ with
parameter $\gamma.$ In the Lagrangian, for a fixed parameter $\gamma,$
the Lagrange parameter $\mu$ is a function of $\gamma.$ Suppose
that $\mu$ is chosen such that the optimum distribution is the generalized
$q$-Gaussian with $\gamma=1.$ In such conditions, we have $J(G_{\gamma})\geq J(G),$
$\forall\gamma.$ On the other hand, we have the following scaling
identities:
\begin{equation}
\begin{cases}
M_{q}[G_{\gamma}]=\gamma^{\frac{n}{\alpha}(q-1)}M_{q}[G],\\
\phi_{\beta,q}[G_{\gamma}]=\gamma^{\frac{\beta\lambda}{\alpha}}\phi_{\beta,q}[G],\\
I_{\beta,q}[G_{\gamma}]=\gamma^{\frac{\beta}{\alpha}}I{}_{\beta,q}[G],\\
Z(\gamma)=\gamma^{-\frac{n}{\alpha}}Z(1)
\end{cases}\label{eq:ScaleEqualities}
\end{equation}
with $\lambda$ given by (\ref{eq:def_lambda}). Thus, the inequality
$J(G_{\gamma})\geq J(G)$ means that 
\begin{equation}
J(G_{\gamma})=\gamma^{\frac{\beta\lambda}{\alpha}}\phi_{\beta,q}[G]+\mu\gamma^{-\frac{2}{\alpha}}N_{q}[G]
\end{equation}
is minimum in $\gamma$ for $\gamma=1$. The derivative with respect
to $\gamma$ then vanishes at $\gamma=1,$ and this yields the value
\begin{equation}
\mu=\mbox{\ensuremath{\frac{\beta\lambda}{2}}}\frac{\phi_{\beta,q}[G]}{N_{q}[G]}.
\end{equation}
 Then, it only remains to set $J_{1}[f]=\mbox{\ensuremath{\frac{N_{q}[G]}{\beta\lambda}\,}}J[f]$
to obtain (\ref{eq:J1_informationfunctional}). The second information
functional (\ref{eq:eq:J2_informationfunctional}) is established
in the same way. 
\end{proof}

\subsection{The minimum of generalized Fisher information among distributions
with a given moment is attained for generalized $q$-Gaussians}

In addition to be the maximum entropy distribution in the set of all
distributions with a given variance, it is well-known that the Gaussian
distribution is also a minimum Fisher information distribution over
all distributions with the same variance. This can be seen as a consequence
of the Cramér-Rao inequality for a location parameter. The same kind
of characterization holds in the multidimensional case for the generalized
$q$-Gaussian and the generalized Fisher information. 

In \cite{lutwak_cramer_2005,lutwak_extensions_2012,bercher__2012},
a Cramér-Rao inequality for generalized $q$-Gaussian involving the
($\beta,q$)-Fisher information $\phi_{\beta,q}[f]$ has been established.
The ($\beta,q$)-Fisher information $I_{\beta,q}[f]$ has been introduced
in the context of parameter estimation in \cite{bercher_generalized_2012,bercher_generalized_2012b}
and a related Cramér-Rao inequality proved. These inequalities mean
that among all distributions with a given moment, the generalized
$q$-Gaussians are the minimizers of extended versions of the Fisher
information, just as the standard Gaussian minimizes Fisher information
over all distributions with a given variance. We begin by recalling
the generalized Cramér-Rao inequalities involving the generalized
Fisher information measures and saturated by the generalized $q$-Gaussians.
Then, we show that the generalized $q$-Gaussians can also be characterized
as the minimizers of two new variational problems. 
\begin{prop}
\label{theo2} {[}Generalized Cramér-Rao inequalities{]} For $n\geq1,$
$\beta$ and $\alpha$ Hölder conjugates of each other, $\alpha>1,$
$q>\max\left\{ (n-1)/n,\, n/(n+\alpha)\right\} $ then for any probability
density $f$ on $\Omega\subseteq\mathbb{R}^{n}$, supposed continuously
differentiable and such that the involved information measures are
finite,
\begin{align}
m_{\alpha}\left[f\right]^{\frac{1}{\alpha}}\,\phi_{\beta,q}\left[f\right]^{\frac{1}{\beta\lambda}} & \geq m_{\alpha}\left[G\right]^{\frac{1}{\alpha}}\,\phi_{\beta,q}\left[G\right]^{\frac{1}{\beta\lambda}}\label{eq:GeneralizedCramerLutwak}\\
m_{\alpha}\left[f\right]^{\frac{1}{\alpha}}\, I_{\beta,q}\left[f\right]^{\frac{1}{\beta}} & \geq m_{\alpha}\left[G\right]^{\frac{1}{\alpha}}\, I_{\beta,q}\left[G\right]^{\frac{1}{\beta}}=n\label{eq:GeneralizedCramerJFB}
\end{align}
with $\lambda=n(q-1)+1$, where the general ($\beta,q$)-Fisher information
are defined in (\ref{eq:Phi_GenFishera}) and (\ref{eq:I_GenFishera})
and where the equality holds iff $f$ is a generalized Gaussian $g=G_{\gamma}.$\end{prop}
\begin{proof}
The proof of (\ref{eq:GeneralizedCramerLutwak}) can be found in \cite{lutwak_extensions_2012}
and \cite{bercher__2012} in the multidimensional case. As far as
(\ref{eq:GeneralizedCramerJFB}) is concerned, the proof is given
in \cite{bercher_generalized_2012,bercher_generalized_2012b} as
a particular case of the more general Cramér-Rao inequality for parameter
estimation. The fact that the lower bound $m_{\alpha}\left[G\right]^{\frac{1}{\alpha}}\, I_{\beta,q}\left[G\right]^{\frac{1}{\beta}}$
is exactly equal to $n$ is a direct consequence of the proof of (\ref{eq:GeneralizedCramerJFB}).

\end{proof}
The Cramér-Rao inequalities (\ref{eq:GeneralizedCramerLutwak})-(\ref{eq:GeneralizedCramerJFB})
imply that $G_{\gamma}$ realizes the minimum of the generalized Fisher
information in the set of all probability densities with a given moment
of order $\alpha:$
\begin{equation}
\inf_{f}\left\{ \phi_{\beta,q}[f]\text{ or }I_{\beta,q}[f]:\, f\in\mathcal{P},\, m_{\alpha}[f]=m_{\alpha}[G_{\gamma}]\right\} ,\label{eq:pb_inf_Fish_m}
\end{equation}
where $\mathcal{P}$ is the probability simplex. As above and as a
direct consequence, we obtain a variational characterization of the
generalized $q$-Gaussians:
\begin{prop}
The generalized $q$-Gaussian $G$ minimize the following two information
functionals
\begin{align}
J_{3}[f] & =\frac{1}{\beta\lambda}\phi_{\beta,q}[f]m_{\alpha}[G]+\frac{1}{\alpha}\phi_{\beta,q}[G]m_{\alpha}[f]\label{eq:J1_informationfunctional-1}\\
J_{4}[f] & =\frac{1}{\beta}I_{\beta,q}[f]m_{\alpha}[G]+\frac{1}{\alpha}I_{\beta,q}[G]m_{\alpha}[f]\label{eq:eq:J2_informationfunctional-1}
\end{align}
\end{prop}
\begin{proof}
The proof proceeds just as the proof of Proposition \ref{Prop2}.
Let us consider the first Cramér-Rao inequality (\ref{eq:GeneralizedCramerLutwak})
and the problem (\ref{eq:pb_inf_Fish_m}).  The corresponding Lagrangian
is 
\begin{equation}
J(f)=\phi_{\beta,q}[f]+\mu m_{\alpha}[f].
\end{equation}
Assume that $\mu$ is chosen such that the optimum distribution is
the generalized $q$-Gaussian with $\gamma=1.$ In such conditions,
we have $J(G_{\gamma})\geq J(G),$ $\forall\gamma.$ By the scaling
identities (\ref{eq:ScaleEqualities}) and the additional fact that
$m_{\alpha}[G_{\gamma}]=\gamma^{-1}m_{\alpha}[G],$ we get that 
\begin{equation}
J(G_{\gamma})=\gamma^{\frac{\beta\lambda}{\alpha}}\phi_{\beta,q}[G]+\mu\gamma^{-1}m_{\alpha}[G].
\end{equation}
Since this functional is minimum for $\gamma=1$, its derivative vanishes
at $\gamma=1,$ which gives 
\begin{equation}
\mu=\mbox{\ensuremath{\frac{\beta\lambda}{\alpha}}}\frac{\phi_{\beta,q}[G]}{m_{\alpha}[G]}
\end{equation}
and in turn the inequality (\ref{eq:J1_informationfunctional-1}).
The second information functional (\ref{eq:eq:J2_informationfunctional-1})
is established in the same way. 
\end{proof}
We shall note that up to the change of function $u(x)^{k}=f(x),$
the Euler-Lagrange equation associated with the two functionals (\ref{eq:J1_informationfunctional-1})
and (\ref{eq:eq:J2_informationfunctional-1}) is a $\beta$-Laplace
equation of the form 
\begin{equation}
\triangle_{\beta}u(x)-\frac{1}{\alpha k^{\beta-1}}\,\frac{\phi_{\beta,q}[G]}{m_{\alpha}[G]}\,|x|^{\alpha}\, u(x)^{k-1}=0,
\end{equation}
where we used the $\beta$-Laplacian operator $\Delta_{\beta}u\,:=\text{div}\left(|\nabla u|^{\beta-2}\,\nabla u\right)$.
In the monodimensional case and if $\beta=2,$ we obtain an instance
of the generalized Emden-Fowler equation $u''(x)+h(x)\, u(x)^{\gamma}=0$,
where $h(x)$ is a given function. This kind of equations arises in
studies of gaseous dynamics in astrophysics, in certain problems in
fluid mechanics and pseudoplastic flow, see \cite{nachman_nonlinear_1980},
as well as in some reaction-diffusion processes.

\subsection{The derivative of $q$-entropies give the generalized Fisher information
-- An extended de Bruijn identity }

A fundamental connection between the Boltzmann-Shannon entropy, Fisher
information, and the Gaussian distribution is given by the de Bruijn
identity \cite{stam_inequalities_1959}. We show here that this important
connection can be extended to the $q$-entropies, the generalized
Fisher information and the generalized $q$-Gaussians. 

 The de Bruijn identity states that if $Y_{t}=X+\sqrt{2t}Z$ where
$Z$ is a standard Gaussian vector and $X$ a random vector of $\mathbb{R}^{n},$
independent of $Z,$ then 
\begin{equation}
\frac{\text{d}}{\text{d}t}H[f_{Y_{t}}]=I_{2,1}[f_{Y_{t}}]=\phi_{2,1}[f_{Y_{t}}],\label{eq:deBruijnIdentity}
\end{equation}
where $f_{Y_{t}}$ denotes the density of $Y_{t}=X+\sqrt{2t}Z$. In
thermodynamics, this can be seen as a consequence of the second law
for an isolated system out of equilibrium. The standard proof of the
de Bruijn identity uses the fact that $Y_{t}$ satisfies the well-known
heat equation \cite{widder_heat_1975} 
\begin{equation}
\frac{\partial f}{\partial t}=\Delta f,\label{eq:he}
\end{equation}
where $\Delta$ denotes the Laplace operator. 

It is possible to consider more general versions of the heat equation.
In particular, it is known that the porous medium equation, or the
fast diffusion equation, admit $q$-Gaussians as solutions. This connects
the $q$-Gaussian distributions to differential equations of physics,
and has been noticed in the nonextensive context, see e.g. \cite{schwaemmle_q-gaussians_2008,ohara_geometric_2009}.
The porous medium equation and fast diffusion equation correspond
to the differential equation 
\begin{equation}
\frac{\partial f}{\partial t}=\Delta f^{m},\label{eq:pme-fpe}
\end{equation}
with $m>1$ for the porous medium equation and $<1$ for the fast
diffusion. These two equations have been exhaustively studied and
characterized by J. L. Vazquez, e.g. in the two books \cite{vazquez_porous_2006,vazquez_smoothing_2006}.
These equations appear in a large number of physical situations, including
fluid mechanics, nonlinear heat transfer or diffusion. Other applications
have been reported in mathematical biology, lubrification, boundary
layer theory, etc, see the series of applications presented in \cite[chapters 2 and 21]{vazquez_porous_2006}
and references therein. 

The porous medium equation (\ref{eq:pme-fpe}) can be generalized
to a wider form, the \textit{doubly nonlinear equation}, which involves
a $p$-Laplacian operator $\Delta_{p}f\,:=\text{div}\left(|\nabla f|^{p-2}\,\nabla f\right),$
and the power $m$ of the porous medium or fast diffusion equation.
This doubly nonlinear equation takes the form
\begin{equation}
\frac{\partial}{\partial t}f=\Delta_{\beta}f^{m}=\text{div}\left(|\nabla f^{m}|^{\beta-2}\,\nabla f^{m}\right),\label{eq:dnle}
\end{equation}
where we use $p=\beta$ for convenience and coherence with the other
notations in the paper. The $\beta$-Laplacian operator typically
in the minimization of a Dirichlet energy like
$\int|\nabla f|^{\beta}\text{d}x$ which leads to the Euler-Lagrange
equation $\Delta_{\beta}f=0.$ A discussion on the $p$-Laplace equation
can be found in \cite{lindqvist_notes_2006}. 

As we can see, the doubly nonlinear equation includes the standard
heat equation ($\beta=2,$ $m=1$), the $\beta$-Laplace equation
($\beta\neq2,$ $m=1$), the porous medium equation ($\beta=2,$ $m>1$)
and the fast diffusion equation ($\beta=2,$ $m<1$). It can be shown,
see \cite[page 192]{vazquez_smoothing_2006}, that for $m(\beta-1)+(\beta/n)-1>0$,
(\ref{eq:dnle}) has a unique self-similar solution, whose initial
value is the Dirac mass at the origin. This fundamental solution is
given by 
\begin{equation}
f(x,t)=\frac{1}{t^{\frac{n}{\delta}}}B\left(\frac{x}{t^{\frac{1}{\delta}}}\right),\,\text{ with }B(x)=\begin{cases}
\left(C-k|x|^{\alpha}\right)_{+}^{\frac{\beta-1}{m(\beta-1)-1}} & \text{ for }m\neq\frac{1}{\beta-1}\\
\frac{1}{\sigma}\exp\left(-\frac{|\beta-1|}{\beta^{\alpha}}|x|^{\alpha}\right) & \text{ for }m=\frac{1}{\beta-1}
\end{cases}\label{eq:FundamentalSolution}
\end{equation}
 with
\begin{equation}
\delta=n(\beta-1)m+\beta-n>0,\,\,\,\,\, k=\frac{m(\beta-1)-1}{\beta}\left(\frac{1}{\delta}\right)^{\frac{1}{\beta-1}}\text{ and }\alpha=\frac{\beta}{\beta-1}.\label{eq:NotationsFundamentalSolution}
\end{equation}
 The constants $C$ and $\sigma$ are uniquely determined by mass
conservation, e.g. $\int f(x,t)\text{d}x=1.$ Of course, we observe
that the function $B(x)$ above is analog to the generalized $q$-Gaussian
(\ref{eq:DefQGaussian}). In the literature, $B(x)$ is called the
Barenblatt profile following its identification as a solution for
heat release from a point source \cite{barenblatt_unsteady_1952,pattle_diffusion_1959}.

The doubly nonlinear diffusion equation also enables to derive a nice
extension of the de Bruijn identity (\ref{eq:deBruijnIdentity}).
A remarkable point is that the extended ($\beta,q$)-Fisher information
naturally pops up in this identity, generalizing the role of the standard
Fisher information in the classical de Bruijn identity. This is stated
in the next Proposition. It shall be mentioned that the case $\beta=2$
of this result has been given in a nice paper by Johnson and Vignat
\cite{johnson_results_2007}. 
\begin{prop}
{[}Extended de Bruijn identity{]} Let $f(x,t)$ a probability distributions
satisfying the doubly nonlinear equation (\ref{eq:dnle}). Assume
that the domain $\Omega$ is independent of $t,$ that $f(x,t)$ is
differentiable with respect to $t,$ and that $\frac{\partial}{\partial t}f(x,t)^{q}$
is absolutely integrable and locally integrable with respect to $t$.
Then, for $q=m+1-\frac{\alpha}{\beta}$, $\lambda=n(q-1)+1,$ and
$\delta=\beta\lambda-n(q-1),$ we have
\begin{align}
\frac{\text{d}}{\text{d}t}H_{q}[f] & =\frac{q\, m^{\beta-1}}{M_{q}[f]}\phi_{\beta,q}[f]=\left(\frac{m}{q}\right)^{\beta-1}M_{q}[f]^{\beta-1}\, I_{\beta,q}[f].\label{eq:ExtendedDeBruijn}\\
\frac{\text{d}}{\text{d}t}S_{q}[f] & =q\, m^{\beta-1}\phi_{\beta,q}[f]=\left(\frac{m}{q}\right)^{\beta-1}M_{q}[f]^{\beta}\, I_{\beta,q}[f].\label{eq:ExtendedDeBruijnb}
\end{align}
where $M_{q}[f]=\int f^{q}$ and where $H_{q}[f]=\frac{1}{1-q}\log M_{q}[f]$
is the Rényi entropy and $S_{q}[f]=\frac{1}{1-q}\left(M_{q}[f]-1\right)$
the Tsallis entropy. In terms of the entropy power $N_{q}[f]=M_{q}[f]^{\frac{2}{n}\,\frac{1}{1-q}}$
we also have
\begin{equation}
\frac{\text{d}}{\text{d}t}N_{q}[f]^{\frac{\delta}{2}}=q\frac{\delta}{n}m^{\beta-1}N_{q}[f]^{\frac{\beta\lambda}{2}}\,\phi_{\beta,q}[f]=\frac{\delta}{n}\left(\frac{m}{q}\right)^{\beta-1}N_{q}[f]^{\frac{\beta}{2}}\, I_{\beta,q}[f],\label{eq:ExtendeddeBruijnEntropyPower}
\end{equation}
and
\begin{equation}
\frac{\text{d}}{\text{d}t}N_{q}[f]^{\frac{\delta}{2}}\geq\frac{\text{d}}{\text{d}t}N_{q}[G]^{\frac{\delta}{2}}.\label{eq:derivativeEntropyPower}
\end{equation}

\end{prop}
Of course, the classical de Bruijn identity 
\begin{equation}
\frac{\text{d}}{\text{d}t}H[f]=\phi_{2,1}[f]=I_{2,1}[f].
\end{equation}
is recovered from the extended de Bruijn identity (\ref{eq:ExtendedDeBruijn}),
for $\alpha=\beta=2,$ and $q=m=1.$ We recognize in the right side
of (\ref{eq:ExtendeddeBruijnEntropyPower}) the very same products
of entropy power and ($\beta,q$)-Fisher information that appear in
the the generalized Stam inequalities (\ref{eq:GeneralizedStamInequality})
and (\ref{eq:GeneralizedStamInequalityforI}). By these inequalities,
we obtain (\ref{eq:derivativeEntropyPower}) which indicates that
the derivative of the entropy power in minimum for the generalized
$q$-Gaussian, with equality if and only if $f=G_{\gamma}$, $\forall\gamma.$
Furthermore, it is known that the solutions of the doubly nonlinear
equation converge to the Barenblatt profile \cite{vazquez_smoothing_2006}.
Thus, we see that the minimum of the derivative of the entropy power
is always attained asymptotically. In a very recent and nice paper
\cite{savare_concavity_2012}, Savaré and Toscani have shown that
in the case $\beta=2,$ $m=q$, the entropy power $N_{q}[f]^{\frac{\delta}{2}}$
is a concave function of $t.$ A consequence of this result is that
the derivative (\ref{eq:ExtendeddeBruijnEntropyPower}) of entropy
power is non increasing in time, with a minimum for $f=G,$ reached
at least asymptotically. It would be very interesting to get the same
kind of result in the present setting. For now, let us proceed with
the proof of extended de Bruijn identity. 
\begin{proof}
Let us consider the Rényi entropy 
\begin{equation}
H_{q}[f]=\frac{1}{1-q}\log\int_{\Omega}f(x,t)^{q}\text{d}x.
\end{equation}
The regularity assumptions in the statement of the proposition enable
to use Leibnitz' rule and differentiate under the integral sign:
\begin{align}
(1-q)\frac{\partial}{\partial t}H_{q}[f] & =\frac{1}{\int_{\Omega}f(x,t)^{q}\text{d}x}\int_{\Omega}\frac{\partial}{\partial t}f(x,t)^{q}\text{d}x=\frac{q}{\int f(x,t)^{q}\text{d}x}\int_{\Omega}f(x,t)^{q-1}\frac{\partial}{\partial t}f(x,t)\text{d}x\\
= & \frac{q}{\int_{\Omega}f(x,t)^{q}\text{d}x}\int_{\Omega}f(x,t)^{q-1}\Delta_{\beta}f(x,t)^{m}\text{d}x\\
= & \frac{q}{\int_{\Omega}f(x,t)^{q}\text{d}x}\int_{\Omega}f(x,t)^{q-1}\nabla.\left[\left|\nabla f^{m}\right|^{\beta-2}\nabla f^{m}\right]\text{d}x
\end{align}
where we used the fact that $f(x,t)$ satisfies the doubly nonlinear
heat equation (\ref{eq:dnle}). By the divergence theorem and the
product rule, we have
\begin{equation}
\int_{U}\,\nabla.(\psi A)\, dV=\int_{\partial U}\psi A.\eta\, dS=\int_{U}\,\psi\nabla.A\, dV+\int_{U}\, A.\nabla\psi\, dV,
\end{equation}
where $\eta$ is a unit vector orthogonal to the hypersurface $\partial U.$
Thus, considering 
\begin{equation}
\int_{\Omega}f(x,t)^{q-1}\Delta_{\beta}f(x,t)^{m}\text{d}x=\int_{\Omega}f(x,t)^{q-1}\nabla.\left[\left|\nabla f^{m}\right|^{\beta-2}\nabla f^{m}\right]\,\text{d}x,
\end{equation}
we get, with $\psi=f(x,t)^{q-1}$ and $A=\left|\nabla f^{m}\right|^{\beta-2}\nabla f^{m}$,
\begin{multline}
\int_{\Omega}f(x,t)^{q-1}\nabla.\left[\left|\nabla f^{m}\right|^{\beta-2}\nabla f^{m}\right]\,\text{d}x=\int_{\partial\Omega}f(x,t)^{q-1}\left|\nabla f^{m}\right|^{\beta-2}\nabla f^{m}.\eta\,\text{d}x\\
-\int_{\Omega}(q-1)f(x,t)^{q-2}\nabla f.\left[\left|\nabla f^{m}\right|^{\beta-2}\nabla f^{m}\right]\,\text{d}x.
\end{multline}
Assuming that both $f(x,t)$ and $|\nabla f|$ tend to zero on the
boundary $\partial\Omega$ of the domain, it remains
\begin{align}
\int_{\Omega}f(x,t)^{q-1}\nabla.\left[\left|\nabla f^{m}\right|^{\beta-2}\nabla f^{m}\right]\,\text{d}x & =-(q-1)\int_{\Omega}f(x,t)^{q-2}\nabla f.\left[\left|\nabla f^{m}\right|^{\beta-2}\nabla f^{m}\right]\,\text{d}x,\\
= & -(q-1)m^{\beta-1}\int_{\Omega}f(x,t)^{\beta m+q-m-1}\left(\frac{\left|\nabla f\right|}{f}\right)^{\beta}\,\text{d}x,
\end{align}
where we used twice the fact that $\nabla f^{m}=mf^{m-1}\nabla f.$
Finally, if we choose $q=m+1-\frac{\alpha}{\beta},$ the exponent
in the last equation reduces to $\beta(q-1)+1,$ and 
\begin{alignat}{1}
\frac{\text{\ensuremath{\partial}}}{\text{\ensuremath{\partial}}t}H_{q}[f] & =\frac{q\, m^{\beta-1}}{\int f(x)^{q}\text{d}x}\int_{\Omega}f(x,t)^{\beta(q-1)+1}\left(\frac{\left|\nabla f\right|}{f}\right)^{\beta}\,\text{d}x\\
 & =\frac{q\, m^{\beta-1}}{M_{q}[f]}\phi_{\beta,q}[f].
\end{alignat}
Using the relationship (\ref{eq:I_GenFishera}) between $I_{\beta,q}\left[f\right]$
and $\phi_{\beta,q}\left[f\right]$, this identity becomes 
\begin{equation}
\frac{\text{\ensuremath{\partial}}}{\text{\ensuremath{\partial}}t}H_{q}[f]=\left(\frac{m}{q}\right)^{\beta-1}M_{q}[f]^{\beta-1}\, I_{\beta,q}[f].
\end{equation}
Of course, the second inequality (\ref{eq:ExtendedDeBruijnb}) follows
immediately. If we consider the entropy power $N_{q}=M_{q}^{\,\frac{2}{n}\,\frac{1}{1-q}}$
instead of the entropy, we obtain along similar steps the following
relation:
\begin{equation}
\frac{\partial}{\partial t}N_{q}[f]^{\frac{\delta}{2}}=q\frac{\delta}{n}m^{\beta-1}N_{q}[f]^{\frac{\beta\lambda}{2}}\,\phi_{\beta,q}[f]=\frac{\delta}{n}\left(\frac{m}{q}\right)^{\beta-1}N_{q}[f]^{\frac{\beta}{2}}\, I_{\beta,q}[f].
\end{equation}
with $\lambda=n(q-1)+1$ and $\delta=n(\beta-1)(q-1)+\beta=\beta\lambda-n(q-1).$

\end{proof}

\section{Thermodynamics considerations }

Given an information measure or a generalized entropy, it is possible
to investigate whether it is possible to construct a thermodynamics
based on this measure. Several of such constructions have been proposed,
namely the nonextensive thermostatistics based on the Tsallis entropy,
see the books \cite{naudts_generalised_2011,tsallis_introduction_2009},
the construction of thermodynamics based on the Rényi entropy, with
connections with the multifractals \cite{bashkirov_renyi_2006,jizba_world_2004,lenzi_statistical_2000,parvan_renyi_2010},
or based on the Sharma-Mittal entropy \cite{frank_generalized_2002}.
Generalized entropies obtained through Beck and Cohen's superstatistics
have also been considered \cite{beck_superstatistics_2003,beck_generalised_2009,tsallis_constructing_2003}.
Finally, following the pioneering work of Frieden, thermodynamics
based on the Fisher information have been proposed \cite{frieden_fisher-based_1999,frieden_physics_2000,frieden_non-equilibrium_2002,frieden_science_2004,plastino_fisher_2005}. 

Though it is not the objective of the present paper to discuss and
compare the pros and cons of these different thermodynamics, these
prior works naturally lead us to examine the status of some basic
properties of conventional thermodynamics in the case of the generalized
Fisher information. In this spirit, we consider the additivity property
of the information, the mixing property, and finally the Legendre
structure of thermodynamics.

\subsection{Additivity}

In classical thermodynamics, additivity is the principle that if a
global system is constituted of two independent components, then the
total entropy is the sum of entropies of each subsystem. However,
the use of a different form of entropy or complexity measure can lead
to reconsider this property. It is well known, for instance, that
the Tsallis entropy is non additive. Indeed, if we use a measure of
information, it is understandable that the information given by a
whole system can be different of the sum of the information given
by its components ; the difference representing an information gain
attached to the combining of the elementary systems. In the case of
Fisher information, the information is additive for independent systems.
However, this property is not preserved for the generalized Fisher
information as soon as $q\neq1$ or $\beta\neq2$. Nevertheless, we
still have an interesting inequality which bounds the ($\beta,q$)-Fisher
information of a combined system. Indeed, if $X$ and $Y$ are two
independent random vectors with densities $f_{X}(x)$ and $f_{Y}(y)$,
then the generalized Fisher information $\phi_{\beta,q}$ attached
to the joint distribution $f_{X,Y}(x,y)=f_{X}(x)f_{Y}(y)$ is
\begin{equation}
\phi_{\beta,q}[f_{X}f_{Y}]=\int\left(f_{X}(x)f_{Y}(y)\right)^{\beta(q-2)+1}\left|\nabla f_{X}(x)\, f_{Y}(y)+f_{X}(x)\nabla f_{Y}(y)\right|^{\beta}\text{d}x\,\text{d}y.\label{eq:genFisherJoint}
\end{equation}
The Minkovski inequality indicates that the triangle inequality holds
in $L_{p}$ spaces, e.g. for a weighted $L_{p}$ norm 
\begin{equation}
\left(\int w(x)|f(x)+g(x)|^{p}\text{d}x\right)^{\frac{1}{p}}\leq\left(\int w(x)|f(x)|^{p}\text{d}x\right)^{\frac{1}{p}}+\left(\int w(x)|g(x)|^{p}\text{d}x\right)^{\frac{1}{p}},
\end{equation}
with $w(x)>0$ and $p>1.$ Applying this to the generalized Fisher
information (\ref{eq:genFisherJoint}), we immediately obtain 
\begin{equation}
\phi_{\beta,q}[f_{X}f_{Y}]^{\frac{1}{\beta}}\leq M_{\beta(q-1)+1}[f_{Y}]^{\frac{1}{\beta}}\phi_{\beta,q}[f_{X}]^{\frac{1}{\beta}}+M_{\beta(q-1)+1}[f_{X}]^{\frac{1}{\beta}}\phi_{\beta,q}[f_{Y}]^{\frac{1}{\beta}},\label{eq:IneqFisherJoint}
\end{equation}
where $M_{q}[f]$ is the information generating function. By the general
power mean inequality, we know that $M_{a}[f]^{\frac{{1}}{a}}\leq M_{b}[f]^{\frac{1}{b}}$
for $a<b.$ Consequently, when $q<1,$ $M_{\beta(q-1)+1}<M_{1}=1,$
and the inequality (\ref{eq:IneqFisherJoint}) reduces to 
\begin{equation}
\phi_{\beta,q}[f_{X}f_{Y}]^{\frac{1}{\beta}}\leq\phi_{\beta,q}[f_{X}]^{\frac{1}{\beta}}+\phi_{\beta,q}[f_{Y}]^{\frac{1}{\beta}}.
\end{equation}

\subsection{Mixing property}

It is well known that the Boltzmann entropy increases when two or
more different substances are mixed. A simple interpretation is that
the uncertainty on (or the complexity of) the mixed system increases
and consequently its the entropy. This mixing property directly implies
that the entropy is a concave function. If we substitute to the classical
entropy a different form which measures differently the amount of
information in a system, then the mixing property, though thermodynamically
appealing, is not necessarily preserved. This is what happens in the
case of our generalized Fisher information. It is known that the standard
Fisher information is a convex function \cite{cohen_fisher_1968,frieden_fisher-based_1999}. 
Similarly, it has be shown that this is also the case for  $q=1$ and any exponent $\beta>1$  \cite{boekee_extension_1977}. We show here that
this is in fact true for any $q\geq 1$. Consider the integrand in the
definition of the ($\beta,q$)-Fisher information, and denote $x=\left|\nabla f\right|$
and $y=f$. Then, let us define 
\begin{equation}
h(x,y)=x^{\beta}y^{(b-\beta)},
\end{equation}
with $b=\beta(q-1)+1.$ This function will be jointly convex, with
respect to its two arguments, if the associated Hessian is non negative
definite. Accordingly, the integral of $h$ will be a convex functional
of $f.$ Hence, we only have to determine the conditions for the non
negativity of the Hessian. 

By direct computation, we readily obtain that the Hessian matrix for
the function $h(x,y)$ is 
\begin{equation}
H=\begin{bmatrix}\dfrac{\partial^{2}h}{\partial x^{2}} & \dfrac{\partial^{2}h}{\partial x\,\partial y}\\[2.2ex]
\dfrac{\partial^{2}h}{\partial y\,\partial x} & \dfrac{\partial^{2}h}{\partial y^{2}}\\[2.2ex]
\end{bmatrix}={x}^{\beta-2}{y}^{b-\beta-2}\left[\begin{array}{cc}
{y}^{2}\beta\left(\beta-1\right) & x\, y\,\beta\left(b-\beta\right)\\
x\, y\,\beta\left(b-\beta\right) & x^{2}\left(b-\beta\right)\left(b-\beta-1\right)
\end{array}\right].
\end{equation}
The determinant of this matrix is
\begin{equation}
D=-{x}^{2\beta-2}{y}^{2(b-\beta)-2}\,\beta\left(b-\beta\right)\left(\, b-1\right),
\end{equation}
and its trace is 
\begin{equation}
T={x}^{\beta-2}{y}^{b-\beta-2}\left(\beta\,{y}^{2}\left(\beta-1\right)+x^{2}\left(b-\beta\right)\left(b-\beta-1\right)\right).
\end{equation}
Since we know that the determinant and trace of a matrix are respectively
the product and sum of the eigenvalues, we see that the two eigenvalues
are non negative if both the determinant and trace are non negative,
for any $x$ and $y.$ This entails here the conditions
\[
\begin{cases}
\beta\left(b-\beta\right)\left(\, b-1\right) & \leq0,\\
\beta\,\left(\beta-1\right) & \geq0,\\
\left(b-\beta\right)\left(b-\beta-1\right) & \geq0.
\end{cases}
\]
Actually, the first two conditions are sufficient since they imply
the third one. With $\beta>0,$ we finally obtain $\beta\geq b\geq1,$
which, given that $b=\beta(q-1)+1,$ finally gives that $\phi_{\beta,q}[f]$ is a convex functional for
\begin{equation}
2-\frac{1}{\beta}\geq q\geq1.
\end{equation}

Let us define by $\phi_{\beta,q}(m)$ and $I_{\beta,q}(m)$ the ``thermodynamic''
($\beta,q$)-Fisher information considered as functions of an observable
$m=E[\left|X\right|^{\alpha}]$:
\begin{equation}
\phi_{\beta,q}(m)=\inf_{f}\left\{ \phi_{\beta,q}[f]:f\in\mathcal{P}\text{ and }m=E[\left|X\right|^{\alpha}]\right\} 
\end{equation}
The value of the information
attached to the optimum distribution is denoted $\phi_{\beta,q}(m)$
-- the use of the square brackets and parenthesis distinguishes between
the functions of the state and the functions of the observable. 
Though the generalized Fisher information $\phi_{\beta,q}[f]$ is
not a convex function of the probability density function for $q<1,$
we should notice that the ``thermodynamic'' ($\beta,q$)-Fisher information
$\phi_{\beta,q}(m)$ and $I_{\beta,q}(m)$ are in fact convex in $m$ for
any $q>\max\left\{ (n-1)/n,\, n/(n+\alpha)\right\} $. This is actually
an immediate consequence of the Cramér-Rao inequalities (\ref{eq:GeneralizedCramerLutwak})
and (\ref{eq:GeneralizedCramerJFB}). Indeed, any generalized $q$-Gaussian
$G_{\gamma}$ with parameter $\gamma$ reaches the Cramér-Rao bound
and we have, for instance for the first ($\beta,q$)-Fisher information: 
\[
m_{\alpha}\left[G_{\gamma}\right]^{\frac{1}{\alpha}}\,\phi_{\beta,q}\left[G_{\gamma}\right]^{\frac{1}{\beta\lambda}}=m_{\alpha}\left[G\right]^{\frac{1}{\alpha}}\,\phi_{\beta,q}\left[G\right]^{\frac{1}{\beta\lambda}}.
\]
Choosing $\gamma$ such that $m=m_{\alpha}\left[G_{\gamma}\right]$,
we obtain that the ``thermodynamic'' ($\beta,q$)-Fisher information
$\phi_{\beta,q}(m)$ is 
\[
\phi_{\beta,q}(m)=K\, m{}^{-\frac{\beta\lambda}{\alpha}},
\]
with $K=m_{\alpha}\left[G\right]^{\frac{\beta\lambda}{\alpha}}\,\phi_{\beta,q}\left[G\right].$
Hence, since $\alpha,$ $\beta$ and $\lambda$ are positive, $\phi_{\beta,q}(m)$
is a convex function of $m$.

\subsection{Reciprocity relations and Legendre thermodynamics relationships}

In the formulation of standard thermodynamics, the Legendre structure
is an important ingredient. We show here that the standard reciprocity relations
and Legendre structure of thermodynamics are still valid for generalized
Fisher information. The fact that the Legendre structure of thermodynamics
is preserved for general maximum-entropy-like problems and is just
a consequence of Jaynes' maximum entropy principle have already be
shown in \cite{plastino_universality_1997}, and re-derived in the
case of the Fisher information in \cite{frieden_fisher-based_1999,plastino_fisher_2005}.
We extend this to the case of the generalized ($\beta,q$)-Fisher
information. Actually, we provide here a simple and self-contained
derivation of the fact that the thermodynamics-like Legendre structure
is a general and direct consequence of a minimization principle. 

We look for the more general distribution, in the Fisher sense, which
is compatible with a prior information given as a mean value of some
observable. This distribution is selected as the distributions with
minimum ($\beta,q$)-Fisher information, and the related variational
problem is stated as
\begin{equation}
I_{\beta,q}(m)=\inf_{f}\left\{ I_{\beta,q}[f]:\, f\in\mathcal{{P}},\, E[A_{i}(X)]=m_{i},\,\, i=1..M\right\} ,\label{eq:MaxEntPbStandardC}
\end{equation}
where $m=[m_1, m_2, \ldots, m_M]$ now denotes the vector of moments $m_{i}=E[A_{i}(X)]\,\, i=1...M.$
The problem consists in finding a distribution with minimum ($\beta,q$)-Fisher
information on the set of all probability distributions with a series
of fixed moments $m_{i},\,\, i=1...M.$ 

Let $\lambda$ be the vector of the Lagrange multipliers $\lambda_{i}$
associated with the $M+1$ constraints, including the constraint $E\left[A_{0}(x)\right]=E\left[1\right]=m_{0}=1$,
and let us define the dual function $D(\lambda)$ by
\begin{equation}
D(\lambda):=\inf_{f\geq0}\left\{ I_{\beta,q}[f]-\sum_{i=0}^{M}\lambda_{i}\left(\int_{\Omega}A_{i}(x)f(x)\text{d}x-m_{i}\right)\right\} .
\end{equation}

It is well known that the maximum of the dual function is always less
than or equal to the value of the primal problem: $\sup_{\lambda}D(\lambda)\leq I_{\beta,q}(m).$
Clearly, the dual function can be expressed in term of the function
$I_{\beta,q}(\lambda)$ defined by 
\begin{equation}
I_{\beta,q}(\lambda):=\sup_{f\geq0}\left\{ \sum_{i=0}^{M}\lambda_{i}\int_{\Omega}A_{i}(x)f(x)\text{d}x-I_{\beta,q}[f]\right\} ,\label{eq:FirstConjugate}
\end{equation}
 and that gives
\begin{equation}
D(\lambda)=\sum_{i=0}^{M}\lambda_{i}m_{i}-I_{\beta,q}(\lambda).
\end{equation}
If we denote by $\bar{A}$ the vector of moments $\bar{A_{i}}=\int_{X}A_{i}(x)f(x)\text{d}x$,
then among all distributions with the same moments $\bar{A}$, there
exists a distribution with minimum ($\beta,q$)-Fisher information,
say $I_{\beta,q}\bar{(A)}$, and this distribution realizes the supremum
in (\ref{eq:FirstConjugate}) for a given moment. Then it remains
to explore all the possible moments $\bar{A}$ and select the one
which maximizes the right hand side of (\ref{eq:FirstConjugate}). 

In other words, if $S_{\bar{A}}$ is the set of distributions with
given moments $S_{\bar{A}}=\left\{ f\geq0:\,\,\bar{A_{i}}=\int_{\Omega}A_{i}(x)f(x)\text{d}x,\, i=0...M\right\} ,$
again with $A_{0}(x)=1$ and $m_{0}=1,$ then
\begin{alignat}{1}
I_{\beta,q}(\lambda) & =\sup_{\bar{A}}\sup_{f\in S_{\bar{A}}}\left\{ \sum_{i=0}^{M}\lambda_{i}\int_{\Omega}A_{i}(x)f(x)\text{d}x-I_{\beta,q}[f]\right\} ,\label{eq:FisherConj2}\\
 & =\sup_{\bar{A}}\left\{ \sum_{i=0}^{M}\lambda_{i}\,\bar{A_{i}}-I_{\beta,q}(\bar{A})\right\} \label{eq:FisherConj3}
\end{alignat}
where $I_{\beta,q}\bar{(A)}$ denotes the ($\beta,q$)-Fisher information
of the distribution which realizes the supremum in (\ref{eq:FirstConjugate})
for a given moment. These simple relations directly imply the so-called
reciprocity relations. Indeed, the maximum in (\ref{eq:FisherConj3})
is attained for a value $\bar{A}_{\lambda}$ such that 
\begin{equation}
\nabla_{\bar{A}}\, I_{\beta,q}(\bar{A}_{\lambda})=\lambda.\label{eq:DiffvsMoy}
\end{equation}
 On the other hand, (\ref{eq:FisherConj3}) gives
\begin{equation}
I_{\beta,q}(\lambda)=\sum_{i=0}^{M}\lambda_{i}\,\bar{A_{i}}_{\lambda}-I_{\beta,q}(\bar{A}_{\lambda}).\label{eq:Ibqlamb}
\end{equation}
Differentiating $I_{\beta,q}(\bar{A}_{\lambda})$ with respect to
$\lambda$ and taking into account (\ref{eq:DiffvsMoy}), we obtain
\begin{equation}
\nabla_{\lambda}\, I_{\beta,q}(\bar{A}_{\lambda})=\left(\nabla_{\lambda}\otimes\bar{A_{\lambda}}\right).\nabla_{\bar{A}}\, I_{\beta,q}(\bar{A}_{\lambda})=\left(\nabla_{\lambda}\otimes\bar{A_{\lambda}}\right).\lambda,\label{eq:EulerTheorem}
\end{equation}
where $x\otimes y$ denotes the outer product of the columns vectors
$x$ and $y$, with $x\otimes y=xy^{t}.$ The result (\ref{eq:EulerTheorem}) is nothing but
a general form of the Euler theorem. Now, the differentiation of $I_{\beta,q}(\lambda)$
in (\ref{eq:Ibqlamb}) with respect to $\lambda$ gives
\begin{equation}
\nabla_{\lambda}\, I_{\beta,q}(\lambda)=\bar{A_{\lambda}}+\left(\nabla_{\lambda}\otimes\bar{A_{\lambda}}\right).\lambda-\left(\nabla_{\lambda}\otimes\bar{A_{\lambda}}\right).\nabla_{\bar{A}}\, I_{\beta,q}(\bar{A}_{\lambda})
\end{equation}
which, taking into account the Euler formula (\ref{eq:EulerTheorem})
yields the simple relation
\begin{equation}
\nabla_{\lambda}\, I_{\beta,q}(\lambda)=\bar{A_{\lambda}}.\label{eq:Diffvslambda}
\end{equation}

In fact, (\ref{eq:FisherConj3}) expresses the fact that
$I_{\beta,q}(\lambda)$ is the Legendre transform of $I_{\beta,q}(\bar{A}),$
and  relations (\ref{eq:DiffvsMoy}) and (\ref{eq:Diffvslambda})
simply express that $\lambda$ and $\bar{A}$ are the dual variables
linked by the induced Legendre structure. However, let us emphasize
that $I_{\beta,q}(\bar{A})$ is not necessarily the Legendre
transform of $I_{\beta,q}(\lambda)$; this will only be true if $I_{\beta,q}(\bar{A})$
is a convex function of $\bar{A}.$ 

Finally, observe that when we select the value of the Lagrange parameter
$\lambda$ which maximizes the dual function, we obtain
\begin{equation}
\nabla_{\lambda}\, I_{\beta,q}(\lambda)=m.
\end{equation}
This gives the way to compute a value of $\lambda$, and therefore
the associated probability distribution, such that the moment constraints
$E[A_{i}]=m_{i},\,\, i=1...M-1$ are satisfied.

\section{Conclusion}

In this paper, we have presented two extended forms of Fisher information
that fit well in the context of nonextensive thermostatistics. Indeed,
these ($\beta,q$)-Fisher information satisfy several extended properties
that interrelate generalized $q$-Gaussians,  generalized ($\beta,q$)-Fisher
information and $q$-entropies. Among these properties, we have
stressed that the generalized $q$-Gaussians are the solution of variational
problems such as the minimization of the generalized ($\beta,q$)-Fisher
information among distributions with a given $q$-entropy, or such
as the minimization of the generalized ($\beta,q$)-Fisher information
among distributions with a fixed moment. This complements the known
fact that the generalized $q$-Gaussians maximize the $q$-entropies
subject to a moment constraint. These results recover, as a particular
case, known characterizations of the standard Gaussian. We have also
introduced several information functionals minimized by the generalized
$q$-Gaussian. In information theory, the important de Bruijn identity
links the Fisher information and the derivative of the entropy. We
have shown that this identity can be extended to generalized
versions of entropy and Fisher information. More precisely, we have
shown that for all distributions satisfying a nonlinear heat equation,
then the generalized ($\beta,q$)-Fisher information naturally pop
up in the expression of the derivative of the entropy. In a second
step, we have examined further properties of the generalized Fisher
information and of their minimization. In particular, we have considered
the combination of two independent systems and shown that, though
non additive, the generalized Fisher information of the combined system
is upper bounded. In the case of mixing, we have shown that the generalized
Fisher information is convex for $q\geq1.$ Finally, we have shown
that the minimization of the generalized Fisher information subject
to moment constraints leads to a Legendre structure analog to the
Legendre structure of thermodynamics. We shall mention that though
these different results were presented using an Euclidean norm on
$\mathbb{R}^{n},$ they are actually also valid for arbitrary norms,
following the general formulation in \cite{bercher_generalized_2012b}.
This means in particular that one can use a weighted norm such as $\left|x\right|=\sqrt{x^{t}C^{-1}x}$,
where $C$ is a covariance matrix, which in turn leads to the characterization
of elliptical generalized $q$-Gaussians. However, this does not directly account
for matrix (or even higher order) constraints, and this still remains an open question 
in the present setting. 

Further work should also examine whether a principle of minimum generalized
($\beta,q$)-Fisher information, together with the underlying Legendre
structure, could lead to interesting predictions. Whatever the outcome,
it is clear that the two generalized Fisher information appear as
useful companions to the $q$-entropies, as demonstrated eg. by de
Bruijn identity. The interplay between the extended $q$-entropies
and generalized $(\beta,\, q)$-Fisher information could certainly
be of value for the analysis and characterization of complex systems,
e.g. following \cite{pennini_rnyi_1998,vignat_analysis_2003,romera_Fisher-Shannon_2004,dehesa_fisher_2006}.

%

\section{\label{sec:Information-measures-of} Appendix - Information measures
of generalized $q$-Gaussians}

Let $g_{\gamma}(x)=\left(1-s\gamma|x|^{\alpha}\right)_{+}^{\frac{\nu}{s}}$.
The corresponding generalized $q$-Gaussian is obtained as $G_{\gamma}(x)=g_{\gamma}(x)/\int g_{\gamma}(x)\text{d}x,$
with $\nu=1$ and $s=q-1$. In order to express the different information
measures of generalized $q$-Gaussians, we use a closed-form expression
given in the next proposition. 
\begin{prop}
Let $\alpha,\, p>0,$ then consider the following quantity: 
\[
\mu_{p,\nu}=\int|x|^{p}\left(1-s\gamma|x|^{\alpha}\right)_{+}^{\frac{\nu}{s}}\text{d}x
\]
We use the change of variable in polar coordinates $x=r\, u,$ with
$u=x/|x|$ and the representation of the Lebesgue measure $\mathrm{d}x=r^{n-1}\mathrm{d}r\,\mathrm{d}u$.
In this expression $\mathrm{d}u$ denotes the surface element on the
unit sphere and $\int\text{d}u=n\,\omega_{n},$ where $\omega_{n}=\pi^{\frac{n}{2}}/\Gamma(\frac{n}{2}+1)$
is the volume of the $n$-dimensional unit ball. By the integral representations
of the Beta function $B$, one gets the formula 
\begin{alignat}{1}
\mu_{p,\nu} & =\frac{1}{\alpha}\,\left(\gamma\right)^{-\frac{p+n}{\alpha}}n\,\omega_{n}\times\nonumber \\
 & \begin{cases}
(-s)^{-\frac{p+n}{\alpha}}B\left(\frac{p+n}{\alpha},-\frac{\nu}{s}-\frac{p+n}{\alpha}\right) & \text{for }-\frac{\nu\alpha}{\left(p+n\right)}<s<0\\
s^{-\frac{p+n}{\alpha}}B\left(\frac{p+n}{\alpha},\frac{\nu}{s}+1\right) & \text{for }s>0\\
\left(\nu\right)^{-\frac{p+n}{\alpha}}\Gamma\left(\frac{p+n}{\alpha}\right) & \text{if }s=0
\end{cases}\label{eq:general_moment_pnu}
\end{alignat}
 
\end{prop}
From this general expression, we immediately identify the partition
function of a generalized $q$-Gaussian $G_{\gamma}$ with parameter
$\gamma$: $Z(\gamma)=\mu_{0,1}$, with $s=q-1$, which yields
\begin{equation}
Z(\gamma)=\frac{1}{\alpha}\left(\gamma\right)^{-\frac{n}{\alpha}}n\,\omega_{n}\times\begin{cases}
(1-q)^{-\frac{n}{\alpha}}B\left(\frac{n}{\alpha},-\frac{1}{q-1}-\frac{n}{\alpha}\right) & \text{for }1-\frac{\alpha}{n}<q<1\\
(q-1)^{-\frac{n}{\alpha}}B\left(\frac{n}{\alpha},\frac{1}{q-1}+1\right) & \text{for }q>1\\
\Gamma\left(\frac{n}{\alpha}\right) & \text{if }q=1.
\end{cases}\label{eq:GenPartitionFunction}
\end{equation}
 Similarly, we obtain the information generating function
\begin{alignat}{1}
M_{q}[G_{\gamma}] & =\int G_{\gamma}(x)^{q}\text{d}x=\frac{\int g_{\gamma}(x)^{q}\text{d}x}{\left(\int g_{\gamma}(x)\text{d}x\right)^{q}}=\frac{\mu_{0,q}}{\left(\mu_{0,1}\right)^{q}},\label{eq:InfoGeneratingFunction}\\
 & =\frac{\alpha q}{n(q-1)+\alpha q}\times\frac{1}{\left(\mu_{0,1}\right)^{q-1}}
\end{alignat}
where the second line is obtained by simplification using the properties
of the Beta functions, namely $B(x,y+1)=\frac{y}{x+y}B(x,y),$ and
$B(x,y-1)=\frac{x+y-1}{y-1}B(x,y)$.

The information generating function, and thus the associated Rényi
and Tsallis entropies are finite for $q>\max{\left\{ n/(n+\alpha),\,1-\alpha/n\right\} }.$

Likewise, the moment of order $p$ is given by 
\begin{equation}
m_{p}[G_{\gamma}]=\frac{\int|x|^{p}g_{\gamma}(x)\text{d}x}{\int g_{\gamma}(x)\text{d}x}=\frac{\mu_{p,1}}{\mu_{0,1}}.
\end{equation}
 If $p=\alpha,$ by the properties of the Beta functions, the expressions
for the moment of order $p$ simplifies into 
\begin{equation}
m_{\alpha}[G_{\gamma}]=\frac{\mu_{\alpha,1}}{\mu_{0,1}}=\frac{n}{\alpha}\frac{1}{\gamma(q-1)\left(\frac{1}{q-1}+\frac{n}{\alpha}+1\right)}\,\text{for }q>n/(n+\alpha)\label{eq:generalized_nu_moment_p=00003D00003D00003Dalpha}
\end{equation}
Let $b=\beta(q-1)+1.$ The generalized ($\beta,q$)-Fisher information
of the generalized Gaussian has the expression 
\begin{alignat}{1}
\phi{}_{\beta,q}[G_{\gamma}] & =\frac{\left(\alpha\gamma\right)^{\beta}}{\left(\mu_{0,1}\right)^{b}}\int|x|^{\alpha}\left(1-(q-1)|x|^{\alpha}\right)_{+}^{\frac{b}{(q-1)}-\beta}\mathrm{d}x\label{eq:general_expression_for_FisherA}
\end{alignat}
which, taking into account that $\beta(\alpha-1)=\alpha$, reduces
easily to 
\begin{equation}
\phi{}_{\beta,q}[G_{\gamma}]=\left(\alpha\gamma\right)^{\beta}\,\frac{\mu_{\alpha,1}}{\mu_{0,1}}\times\frac{1}{\left(\mu_{0,1}\right)^{\beta(q-1)}}.
\end{equation}
 The second generalized ($\beta,q$)-Fisher information is given by
\[
I{}_{\beta,q}[G_{\gamma}]=\frac{\phi{}_{\beta,q}[G_{\gamma}]}{M{}_{q}[G_{\gamma}]^{\beta}}=\left(\alpha\gamma\right)^{\beta}\,\frac{\mu_{\alpha,1}}{\mu_{0,1}}\times\left(\frac{\mu_{0,1}}{\mu_{0,q}}\right)^{\beta}.
\]
Let us finally evaluate the lower bound in the Cramér-Rao inequality
(\ref{eq:GeneralizedCramerJFB}). We have 
\[
I_{\beta,q}\left[G_{\gamma}\right]^{\frac{1}{\beta}}\, m_{\alpha}\left[G_{\gamma}\right]^{\frac{1}{\alpha}}=\alpha\gamma\,\left(\frac{\mu_{0,1}}{\mu_{0,q}}\right)\left(\frac{\mu_{\alpha,1}}{\mu_{0,1}}\right)^{\frac{1}{\beta}}\left(\frac{\mu_{\alpha,1}}{\mu_{0,1}}\right)^{\frac{1}{\alpha}}=\alpha\gamma\,\left(\frac{\mu_{\alpha,1}}{\mu_{0,q}}\right),
\]
since $\alpha^{-1}+\beta^{-1}=1.$ By properties of the Beta functions,
we obtain that the ratio simplifies to $\mu_{\alpha,1}/\mu_{0,q}=n/(\gamma\alpha)$.
Therefore, we get the equality $I_{\beta,q}\left[G_{\gamma}\right]^{\frac{1}{\beta}}\, m_{\alpha}\left[G_{\gamma}\right]^{\frac{1}{\alpha}}=I_{\beta,q}\left[G\right]^{\frac{1}{\beta}}\, m_{\alpha}\left[G\right]^{\frac{1}{\alpha}}=n$,
which is the lower bound in (\ref{eq:GeneralizedCramerJFB}).


\end{document}